\documentclass[10pt,conference]{IEEEtran}
\usepackage{cite,graphicx,psfrag,amsmath,amssymb,subfigure}
\newtheorem{theorem}{Theorem}

\def\Tiny{\fontsize{6pt}{6pt}\selectfont}

\begin{document}
\title{Improved Redundancy Bounds for Exponential Objectives}
\bibliographystyle{IEEEtran}
\author{\authorblockN{Michael~B.~Baer}
\authorblockA{Vista Research\\
Monterey, CA, USA\\
Email: calbear\hspace{1sp}@\hspace{1sp}ieee.org}}
\maketitle

\begin{abstract}
We present new lower and upper bounds for the compression rate of binary prefix codes optimized over memoryless sources according to two related exponential codeword length objectives.  The objectives explored here are exponential-average length and exponential-average redundancy.  The first of these relates to various problems involving queueing, uncertainty, and lossless communications, and it can be reduced to the second, which has properties more amenable to analysis.  These bounds, some of which are tight, are in terms of a form of entropy and/or the probability of an input symbol, improving on recently discovered bounds of similar form.  We also observe properties of optimal codes over the exponential-average redundancy utility.
\end{abstract}

\section{Introduction}
\label{intro}

Among Shannon's many observations in the seminal paper on information theory was that, by increasing block size, the compression rate of a block code for a memoryless source can get arbitrarily close to the source entropy rate.  In particular, given a block of Shannon entropy $H$ bits, prefix coding methods such as Huffman coding can code the block with an expected length $L$, where $L \in [H,H+1)$.  If $p_i \in (0,1)$ is the probability of the $i$th item, which has a codeword of length $l_i$, then
$$L \triangleq \sum_i p_il_i \mbox{  and  } H \triangleq - \sum_i p_i\lg p_i$$
where $\lg \triangleq \log_2$ and the sum is, without loss of generality, taken over the $n$ possible items.  A constant absolute difference translates into an arbitrarily close-to-entropy compression ratio as blocks grow in size without bound.  The lower bound is fundamental to the definition of entropy, while the upper bound is easily seen by observing the suboptimal Shannon code.  This code, that in which an event of probability $p$ is coded into a codeword of length $\lceil -\lg p \rceil$, will always have expected length less than $H+1$ and never have expected length less than~$L$.

This \textit{unit-sized bound} is preserved even for many nonlinear optimization criteria.  Such criteria are encountered in a variety of lossless compression problems in which expected length is no longer the value to minimize.  In particular, consider
\begin{equation}
L_a = L_a({\mbox{\boldmath $l$}},p)
\triangleq \log_a \sum_{i=1}^n p_ia^{l_i}.
\label{one}
\end{equation}
Minimizing this utility solves several problems involving compression for queueing\cite{Humb2}, compression with uncertainty\cite{ReCh}, one-shot communications\cite{Baer07}, and unreliable communications\cite{Baer11}.  It is closely related to R\'{e}nyi entropy 
\begin{equation}
H_\alpha(p) \triangleq \frac{1}{1-\alpha}\lg \sum_{i=1}^n p_i^\alpha
\label{entropy}
\end{equation}
in the sense that, for $\alpha = 1/(1+\lg a)$,
$$H_\alpha(p) \leq L_a^{\mathop{\rm opt}} < H_\alpha(p)+1.$$
Limits define R\'{e}nyi entropy for $0$, $1$, and $\infty$, so that
$$H_0(p) \triangleq \lim_{\alpha \downarrow 0} H_\alpha(p) = \lg
\|p\|$$ (the logarithm of the number of events in $p$), $$H_1(p) \triangleq
\lim_{\alpha \rightarrow 1} H_\alpha(p) = - \sum_{i=1}^n p_i \lg
p_i$$ (the Shannon entropy of $p$), and $$H_\infty(p) \triangleq
\lim_{\alpha \uparrow \infty} H_\alpha(p) = -\lg \max_i p_i$$ (the
min-entropy).  Over a constant $p$, entropy is nonincreasing over~$\alpha$\cite{Ren1}.

$L_a$ is also closely related to \textit{exponential-average redundancy} or \textit{exponential redundancy}
$$R^d(p,{\mbox{\boldmath $l$}}) \triangleq \frac{1}{d} \lg \sum_i p_i^{1+d} 2^{dl_{i}} =
\frac{1}{d} \lg \sum_i p_i 2^{d(l_{i}+\lg p_i)}.$$
If we substitute $d =\lg a$ and 
$$\hat{p}_i \triangleq \frac{p_i^{{\alpha}}}{\sum_{k=1}^n p_k^{{\alpha}}} = 
\frac{p_i^{{\alpha}}}{2^{(1-{\alpha})H_{{\alpha}}(p)}}$$
we find
\begin{equation}
\begin{array}{rcl}
R^{\lg a}(\hat{p},{\mbox{\boldmath $l$}}) 
&=& \displaystyle \frac{1}{\lg a} \lg \sum_{i=1}^n {\hat{p}_i}^{1+\lg a} a^{l_i} \\
&=& \displaystyle \log_a \sum_{i=1}^n p_ia^{l_i} - 
\log_a \left(\sum_{i=1}^n p_i^{{\alpha}}\right)^{\frac{1}{{\alpha}}} \\
&=& \displaystyle L_a({\mbox{\boldmath $l$}},p) - H_{{\alpha}}(p).
\end{array}
\label{trans}
\end{equation}
This transformation --- shown previously in \cite{BlMc} --- provides a reduction from $L_a$ to $R^d$, allowing bounds for the former to apply --- with the addition of the entropy term --- to the latter.

For both the traditional and exponential utilities, we can improve on the unit-sized bound given the probability of one of the source events.  This was first done with the constraint that the given probability be the most probable of these events\cite{Gall}, but here, as in some subsequent work\cite{YeYe2,MPK,Baer11}, we drop this constraint.  Without loss of generality, we call the source symbols $\{1,2,\ldots,n\} = {\mathcal X}$ (from most to least probable), and call the symbol with known probability~$j$; that is, $p_j$ is known, but not necessarily $j$ itself.

In traditional linear optimization, upper and lower bounds for $R^d$ are known such that probability distributions can be found achieving or approaching these bounds\cite{YeYe2,MPK}; i.e., they are \textit{tight}.  In the exponential cases, \cite{Baer11} took $a \uparrow \infty$ ($d \uparrow \infty$) and $a \downarrow 1$ ($d \downarrow 0$), using inequality relations to find not-necessarily-tight bounds on these problems in terms of tight bounds for the limit cases.  The goal here is to improve the bounds.

We seek to find an upper bound $\omega^d(p_j)$ and lower bound $o^d(p_j)$ such that, for every probability distribution $p$, optimal codeword lengths ${\mbox{\boldmath $l$}}$ satisfy:
$$0 \leq o^d(p_j) \leq \min_{\mbox{\boldmath \scriptsize $l$}} R^d(p,{\mbox{\boldmath $l$}}) < \omega^d(p_j) \leq 1$$
for any~$j$.  For such values, (\ref{trans}) results in:
\begin{eqnarray*}
o^{\log a}\left(p_j^{\tilde{\alpha}} 2^{(\tilde{\alpha}-1)H_{\tilde{\alpha}}(p)}\right) 
&\leq& L_a^{\mathop{\rm opt}}(p) - H_{\tilde{\alpha}}(p) \\
&<& \omega^{\log a}\left(p_j^{\tilde{\alpha}} 2^{(\tilde{\alpha}-1)H_{\tilde{\alpha}}(p)}\right)
\end{eqnarray*}
where $\tilde{\alpha} = 1/(1+\lg a)$ and $L_a^{\mathop{\rm opt}}(p)$ denotes the utility for optimal lengths given $p$ and~$a$.  Thus we can restrict ourselves to exponential redundancy, which is more amenable to the analysis used here.

\section{Applications}

\subsection{$d>0$ ($a>1$)}

Most applications of the exponential length utility concern only $a>1$ ($d>0$ for the redundancy equivalent).  The first known application, introduced in Humblet's dissertation\cite{Humb0,Humb2}, is in a queueing problem originally posed by Jelinek\cite{Jeli}.  Codewords coding a random source are temporarily stored in a finite buffer; these are chosen such that overflow probability is minimized.

Another application considers a source with uncertain probabilities, one in which we only know that the relative entropy between the actual probability mass function and $p$ is within a known bound\cite{ReCh}.  A third, more recent application, omitted in the interest of brevity but described in \cite{Baer11}, is a modified case of the application in the next paragraph.

\subsection{$d<0$ ($a<1$)}

An application for $a<1$ involves single-shot communications with a communication channel having a window of opportunity of geometrically-distributed length (in bits)\cite{Baer07}.  If the distribution has parameter $a$, the probability of successful
transmission is
$${\mathbb P}[\mbox{success}] = a^{L_a(p,{\mbox{\boldmath \scriptsize $l$}})} = \sum_{i=1}^n p_i a^{l_i}.$$
Maximizing this is equivalent to minimizing (\ref{one}).  The solution is trivial for $a \leq 0.5$ ($d \leq -1$), a case not covered by R\'{e}nyi entropy, and thus not applicable here.

\section{Bounds}

The variation of the Huffman algorithm which finds an optimal code for exponential redundancy differs as follows:  While Huffman coding inductively pairs the two lowest probabilities (weights) $w_x$ and $w_y$, combining them into an item \textit{weighted} $f(w_x,w_y) \triangleq w_x + w_y$, optimizing exponential redundancy requires the combined item to be weight
\begin{equation}
f^d(w_x,w_y) \triangleq\left(2^d w_x^{1+d}+2^d w_y^{1+d}\right)^{\frac{1}{1+d}} .
\label{mmprcomb}
\end{equation}
The optimality of this is shown in \cite{Park} and can illustrated with an exchange argument (e.g., \cite[pp.~124-125]{CoTh2} for the linear case).  An exchange argument also inductively illustrates that such an algorithm, depending on how ties are broken, can achieve \textit{any} optimal set of codeword lengths: Clearly the only optimal code is obtained for $n=2$.  Let $n'$ be the smallest $n$ for which there is a set of $\{l_i\}$ that is optimal but cannot be obtained via the algorithm.  Since $\{l_i\}$ is optimal, consider the two smallest probabilities, $p_{n'}$ and $p_{n'-1}$.  In this optimal code, two items having these probabilities (although not necessarily items $n'-1$ and $n'$) must have the longest codewords and must have the same codeword lengths.  Otherwise, we could exchange the codeword with a longer codeword corresponding to a more probable item and improve the utility function, showing nonoptimality.  Merge these two items into one with probability $f^d(p_{n'}, p_{n'-1})$, as per the algorithm.  Because of the nature of $f^d$, this is a reduced problem, i.e., an equivalent optimization to the original problem.  This means that there is a set of lengths optimal for this problem such that all non-merged items are identical to the corresponding $l_i$, while the merged item is simply one shorter than the longest~$l_i$.  Since we inductively assumed all optimal length sets could be produced for $n'-1$, the assumption is verified for all~$n$.

Related observations form the following theorem, similar to that in \cite{Baer11} for a non-exponential utility:

\begin{theorem}
Suppose we apply (\ref{mmprcomb}) to find a Huffman-like code tree in order to minimize exponential redundancy $R^d(p,{\mbox{\boldmath $l$}})$ for $d > -1$.  Then the following holds for any optimal~${\mbox{\boldmath $l$}}$:
\begin{enumerate}
\item For $d>0$, items are always merged by nondecreasing weight and the total probability of any subtree is no greater than the weight of the (root of the) subtree.  For $d<0$, the total probability of any subtree is no less than the weight of the subtree.
\item The weight of the root of the coding tree is $w_{\mathop{\rm root}} = 2^{R^d(p,{\mbox{\scriptsize \boldmath $l$}})}$.
\item If $p_1 \leq f^d(p_{n-1},p_n)$, then an optimal code can be represented 
by a \textit{complete tree}, that is, a tree
with leaves at depth $\lfloor \lg n \rfloor$ and $\lceil \lg n \rceil$
only (with $\sum_{i} 2^{-l_{i}} = 1$).
\end{enumerate}
\label{complete}
\end{theorem}

\begin{proof}
Again we use induction, this time using trivial base cases of sizes $1$ and $2$, and assuming the propositions true for sizes $n-1$ and smaller.  We assume without loss of generality that, for size $n$, items $n-1$ and $n$ are the first to be merged.  We use weight terminology ($w$) instead of probabilities ($p$) because reduced problems need not have weights sum to~$1$.

The subtree part of the first property considers subtrees of size $n$, not necessarily the whole coding tree.  All we need to have a successful reduction to size $n-1$ is to show the following:
\begin{eqnarray}
f^d(w_x,w_y) &=& \left(2^d w_x^{1+d}+2^d w_y^{1+d}\right)^{\frac{1}{1+d}} \label{lhs} \\
&\geq& w_x + w_y\label{rhs}
\end{eqnarray}
for $d > 0$, and 
\begin{equation}
f^d(w_x,w_y) \leq w_x + w_y
\label{rhs2}
\end{equation}
for $d \in (-1,0)$, with equality in either case if and only if $w_x = w_y$.  The inequalities are due to the identical property of the generalized mean in \cite[3.2.4]{AbSt}:
$$M(t)=\left(\frac{1}{m}\sum_{k=1}^m a_k^t\right)^\frac{1}{t}$$
with, in this case, $m=2$, $a_1 = 2w_x$, $a_2 = 2w_y$, and $t$ as $1+d$ in (\ref{lhs}) (left-hand side of (\ref{rhs2})) and $1$ on (\ref{rhs}) (right-hand side of~(\ref{rhs2})).

It immediately follows in the $d>0$ case that $f^d(w_x,w_y) > w_x$.  Thus, the first two weights of the entire tree merge form a weight no less than either original weight, and all remaining weights are also no less that those two weights.  Call the resulting lengths~${\mbox{\boldmath $l$}}'$.

To prove the second property, note that, after merging the aforementioned two least weighted items, we have $n-1$ weights, and thus a conforming reduced problem.  Call the combined weight $w'_{\mathop{\rm c}}$.  Then
\begin{eqnarray*}
w_{\mathop{\rm root}} &=& 2^{R^d(p,{\mbox{\boldmath \scriptsize $l$}})} \\
&=&\left({w'_{\mathop{\rm c}}}^{1+d}2^{(l_n-1)d} + \sum_{i=1}^{n-2} p_i^{1+d}2^{l_id} \right)^\frac{1}{d} \\
&=& \left(p_{n-1}^{1+d}2^{l_{n-1}d}+p_n^{1+d}2^{l_nd} + \sum_{i=1}^{n-2} p_i^{1+d}2^{l_id}\right)^\frac{1}{d} \\
&=& 2^{R^d(p,{\mbox{\boldmath \scriptsize $l$}})}
\end{eqnarray*}
where the third equality is due to $l_{n-1} = l_n$ and (\ref{mmprcomb}).

The third property is shown via the operation of the algorithm from
start to finish: First note that $\sum_i 2^{-l_i} = 1$ for any tree
created using the Huffman-like procedure, since all internal nodes
have two children.  Now think of the procedure as starting with a
priority queue of input items, ordered by nondecreasing weight from
head to tail.  After merging two items, obtained from the head, into
one compound item, that item is placed back into the queue.  Since we
are using a priority queue, the merged item is placed such that its
weight is no smaller than any item ahead of it and is smaller than any
item behind it.

In keeping items ordered, we obtain an optimal coding tree.  A first
derivative test shows that $f^d$ is nondecreasing on both inputs for
any $d$.  Thus merged items are created in nondecreasing weight.  If
$p_1 \leq f^d(p_{n-1},p_n)$, the first merged item can be inserted to
the tail of the queue; since merged items are created in nondecreasing
weight, subsequent items are as well.  This is a sufficient condition for a complete tree being optimal\cite[Lemma~2]{Baer07}.
\end{proof}

\psfrag{Rlb}{\scriptsize $R_{\mathop{\rm opt}}^d(p)$}
\psfrag{Rub}{\scriptsize $R_{\mathop{\rm opt}}^d(p)$}
\psfrag{pj}{\scriptsize $p_j$}
\psfrag{inf  }{\Tiny $\infty$}
\psfrag{1l   }{\Tiny $~~1$}
\psfrag{2l   }{\Tiny $~~2$}
\psfrag{4l   }{\Tiny $~~4$}
\psfrag{16l  }{\Tiny $16$}
\psfrag{-0.5l}{\Tiny $\!\!\!\!\!\!-0.5$}
\psfrag{-1l  }{\Tiny $\!\!-1$}

\begin{figure*}[ht]
     \centering
\psfrag{0l   }{\Tiny $\!\!0^-$}
     \subfigure[Upper bounds]
     {\includegraphics[width=.48\textwidth]{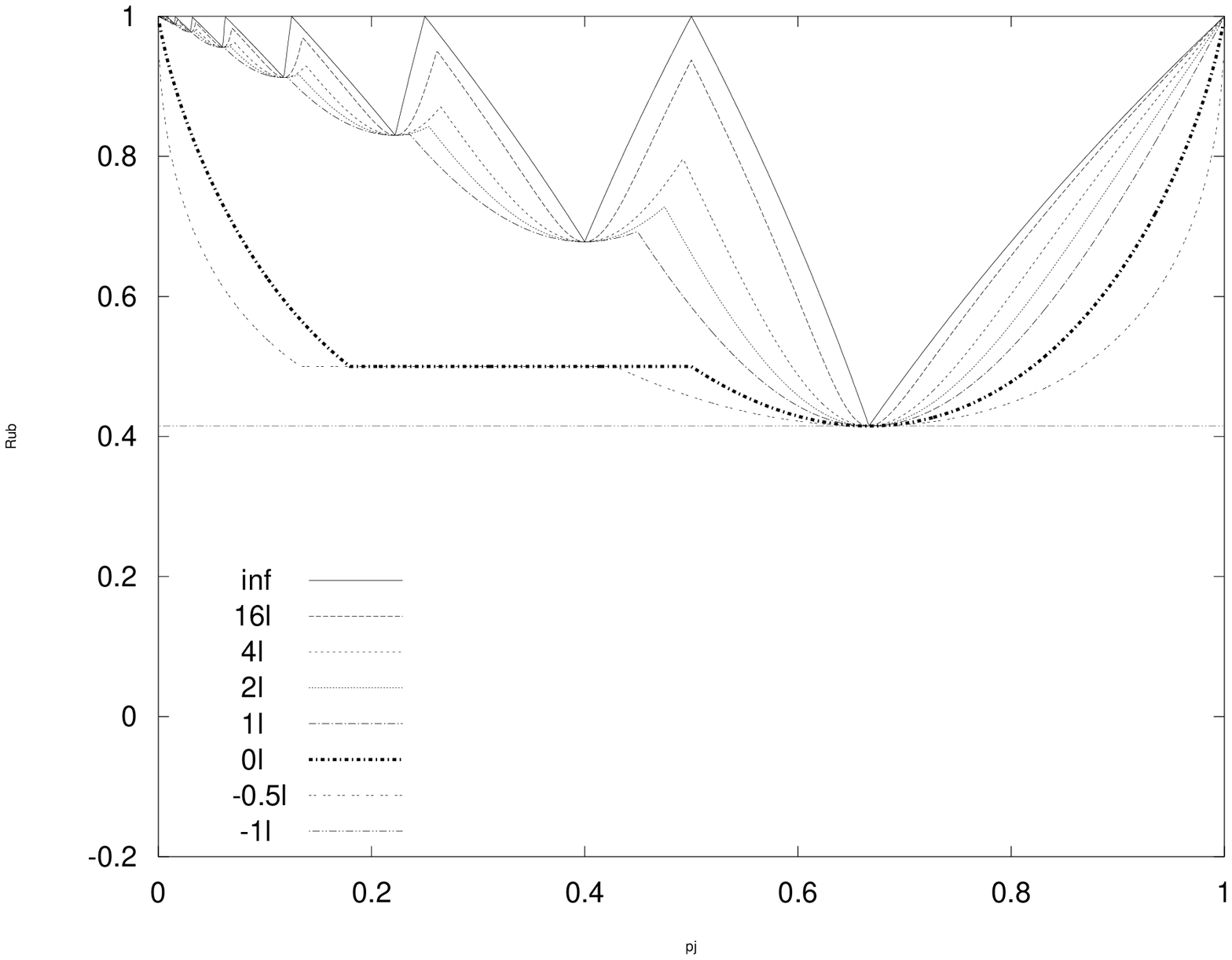}}
\psfrag{0l   }{\Tiny $~~0$}
     \subfigure[Lower bounds]
     {\includegraphics[width=.48\textwidth]{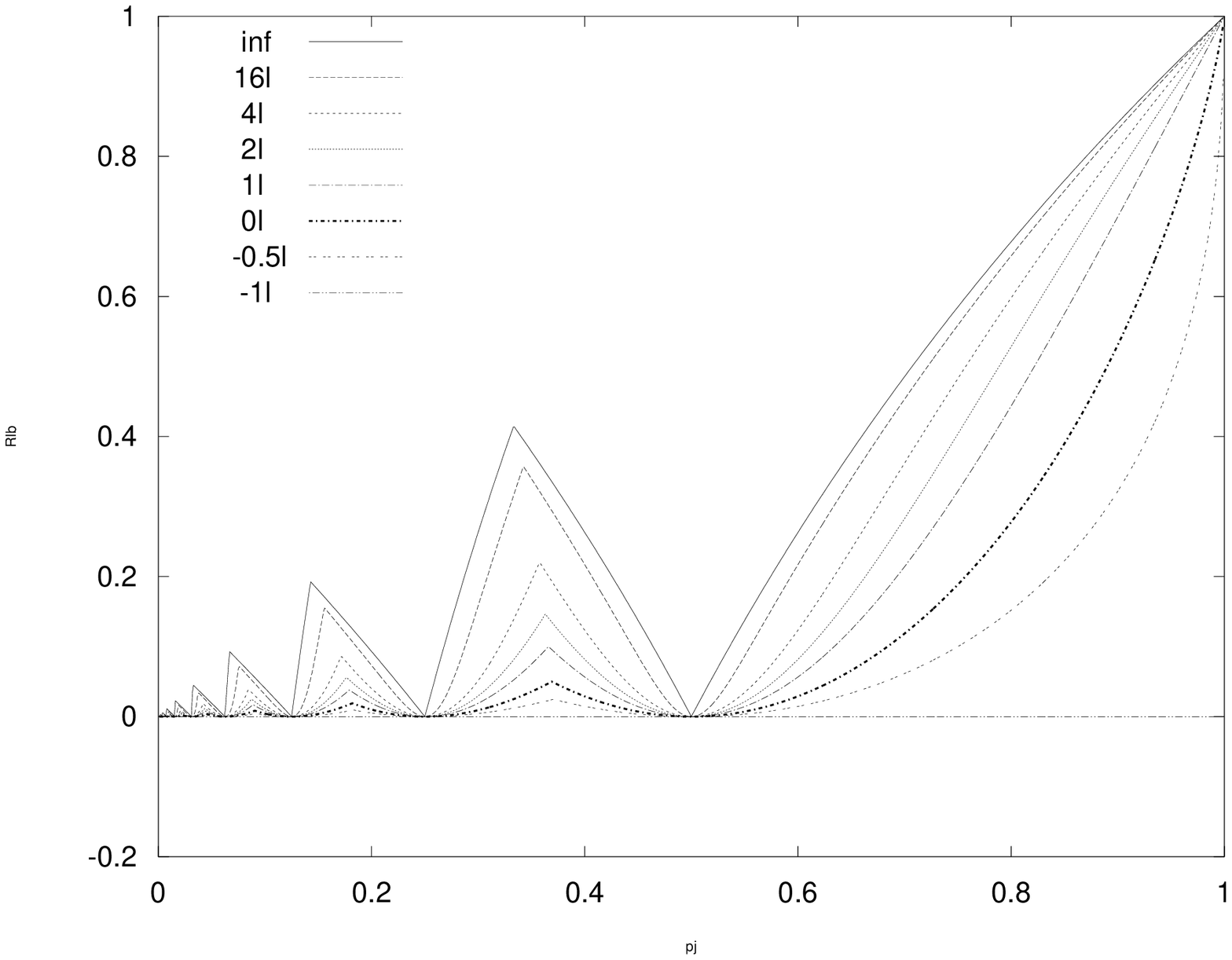}}
     \caption{Bounds on optimal $R_{\mathop{\rm opt}}^d(p)$ given $p_j$ over various $d$ (see legends).  The thick (dash-dotted) lines correspond to the usual linear redundancy utility ($d \rightarrow 0$), while the uppermost (solid) lines are minimum maximum pointwise redundancy ($d \rightarrow \infty$).  Lower bounds are tight over all $d>-1$, while upper bounds are only tight for minimum maximum pointwise redundancy, for $p_j \geq 0.5$ if $d \in (-1,\infty)$, and for $(0,\pi_0^d)$ if $d \in (-1,0)$, where $\pi_0^d$ as the first root of the equality of the two terms in the maximization at (\ref{tightupper}).  The tight upper bounds for $d < \infty$ are approached by $p = (p_j, 1-p_j-\epsilon, \epsilon)$.}
     \label{fig}
\end{figure*}
Next is our main result:
\begin{theorem}
Suppose we know $d>-1$ ($d \neq 0$) and one $p_j$ of probability mass function $p$ for which we want to find the optimal code ${\mbox{\boldmath $l$}}$ under exponential redundancy.  Consider functions

\begin{equation}
\omega^d(p_j) = \min_{\lambda \in {\mathbb Z}^+} \left(\lambda + \frac{1}{d} \lg\left(p_j^{1+d} + \frac{2^d(1-p_j)^{1+d}}{(2^{\lambda}-1)^d}\right)\right)
\label{p1b1}
\end{equation}
making transitions between $\lambda$ and $\lambda+1$ at 
$$\textstyle p_\lambda = \left(1+\left(
\left(1-2^{-d}\right)\left(\frac{1}{(2^\lambda-1)^d}-\frac{1}{(2^\lambda-0.5)^d}\right)^{-1}\right)^{\frac{1}{1+d}}\right)^{-1}$$
and
$$o^d(p_j) = \min_{\mu \in {\mathbb Z}^+} \left(\mu + \frac{1}{d}\lg\left(p_j^{1+d} + \frac{(1-p)^{1+d}}{(2^\mu-1)^d}\right)\right)$$
with transitions between $\mu$ and $\mu+1$ at
$$\textstyle p_\mu = \left(1+\left(
\left(2^d-1\right)\left(\frac{1}{(2^\mu-1)^d}-\frac{1}{(2^\mu-0.5)^d}\right)^{-1}\right)^{\frac{1}{1+d}}\right)^{-1}$$
These improve bounds on the optimal code, and the upper bound is a strict inequality, in that $$0 \leq o^d(p_j) \leq R^d(p,{\mbox{\boldmath $l$}}) < \omega^d(p_j) \leq 1.$$  Moreover, the lower bounds are achievable given $p_1$ and the upper bounds are approachable given $p_1 \geq 0.5$.  In addition, for $p_j < 0.5$ and $d<0$, we have the following secondary upper bound:
\begin{equation}
R^d(p,{\mbox{\boldmath $l$}}) < \max\left(0.5, \frac{1}{d} \lg \left(p_j^{1+d}4^d+(1-p_j)^{1+d}2^d\right)\right).
\label{tightupper}
\end{equation}
\end{theorem}

\begin{proof}

\subsubsection{Lower bound}

The lower bound calculation is:
\begin{eqnarray*}
R^d(p,{\mbox{\boldmath $l$}}) &=& \frac{1}{d} \lg \sum_{i \in {\mathcal X}} p_i^{1+d} 2^{dl_i} \\
&=& \frac{1}{d} \lg \bigg(p_j^{1+d}2^{dl_j} + (1-p_j)^{1+d} 2^{dl_n} \cdot\\
&&\quad \sum_{i \in {\mathcal X} \backslash \{j\}} 2^{l_n-l_i} \left(\frac{p_i2^{l_i-l_n}}{1-p_j}\right)^{1+d} \bigg)\\
&\stackrel{{\mbox{\footnotesize (a)}}}{=}{}& \frac{1}{d} \lg \bigg(p_j^{1+d}2^{dl_j} + (1-p_j)^{1+d}2^{dl_n}\cdot\\
&&\quad \sum_{i \in {\mathcal X} \backslash \{j\}} \sum_{k=1}^{2^{l_n-l_i}} \left(\frac{p_i2^{l_i-l_n}}{1-p_j}\right)^{1+d} \bigg)\\
&\stackrel{{\mbox{\footnotesize (b)}}}{\geq}{}& \frac{1}{d} \lg \bigg(p_j^{1+d}2^{dl_j} + \\
&&\quad (1-p_j)^{1+d}2^{dl_n} \left(2^{l_n}-2^{l_n-l_j}\right)^{-d}\bigg)\\
&=& l_j + \frac{1}{d}\lg \left(p_j^{1+d} + (1-p_j)^{1+d} \left(2^{l_j}-1\right)^{-d}\right)
\end{eqnarray*}
The first equality is due to the definition, while the other equalities follow from algebra.  The summation following (a) is a sum of the $(1+d)$ power of $2^{l_n} - 2^{l_n-l_j}$ positive terms which sum to~$1$.  Consider these values, which include $2^{l_n-l_i}$ repetitions of each $p_i2^{l_i-l_n}/(1-p_j)$ for $i \neq j$, as a probability distribution called~$q$.  Then the summation is related to the $(1+d)$-R\'{e}nyi entropy of $q$; substituting using its definition (\ref{entropy}) leads to (\ref{qentropy}) below.  Furthermore, because $H_0(q) = \lg \|q\|$ and $H_\alpha$ is nonincreasing with $\alpha$, (\ref{qentropy}) is bounded as follows:
\begin{eqnarray}
\left(\sum_{m=1}^{2^{l_n} - 2^{l_n-l_j}} q_m^{1+d}\right)^{\frac{1}{d}} &=& 2^{-H_{1+d}(q)}.\label{qentropy}\\
&\geq&2^{-\lg\|q\|} = (2^{l_n} - 2^{l_n-l_j})^{-1}. \nonumber
\end{eqnarray}
This results in inequality (b), completing the lower bound by substituting minimizing $\mu$ for~$l_j$.  The transitions follow from algebraically finding where there are two minimizing values.

A code achieving this lower bound, for $p_1 = p_j \in [1/(2^{\mu+1}-1),1/2^{\mu})$ for some $\mu$, is $$\left(p_1, \underbrace{\frac{1-p_1}{2^{\mu+1}-2}, \ldots,
\frac{1-p_1}{2^{\mu+1}-2}}_{2^{\mu+1}-2}\right).$$ By
Theorem~\ref{complete}, this has a complete coding tree --- recall $f^d(w_x,w_x) =2w_x$ --- in this case with $l_1$ one bit shorter than the other lengths.  This is easily calculated as achieving the lower bound.

\subsubsection{Upper bounds}

Consider the following code for an arbitrary $\lambda$, as in \cite{YeYe2}:

$$
l_i^j(p) = \left\{
\begin{array}{ll}
\lambda,& i = j \\ 
\left\lceil -\lg
\left(p_i\left(\frac{1-2^{-\lambda}}{1-p_j}\right)\right) \right\rceil,& i \neq j
\end{array}
\right.
$$
Satisfying the Kraft inequality, it is a valid --- possibly suboptimal --- code, and thus has a utility that upper-bounds that of the optimal code.  Thus:

\begin{eqnarray*}
R^d(p,{\mbox{\boldmath $l$}}) &=& \frac{1}{d} \lg \sum_{i \in {\mathcal X}} p_i^{1+d} 2^{dl_i} \\
&\leq& \frac{1}{d} \lg \Bigg(p_j^{1+d} 2^{d\lambda} + \\
&&\quad \sum_{i \in {\mathcal X} \backslash \{j\}} p_i^{1+d} 2^{d\left\lceil -\lg \left(p_i(1-2^{-\lambda})/(1-p_j)\right) \right\rceil}\Bigg)\\
&<& \frac{1}{d} \lg \Bigg(p_j^{1+d} 2^{d\lambda} + \\
&&\quad \sum_{i \in {\mathcal X} \backslash \{j\}} p_i^{1+d} \left(\frac{p_i}2 \cdot \frac{1-2^{-\lambda}}{1-p_j}\right)^{-d} \Bigg)\\
&=& \frac{1}{d} \lg\left(p_j^{1+d}2^{d\lambda} + (1-p_j)^{1+d}\left(\frac{2}{1-2^{-\lambda}}\right)^d\right)
\end{eqnarray*}
Since $\lambda$ is arbitrary, the bound is obtained by choosing the value offering the strictest bound.  This upper bound is approached for any $d>-1$ over $p_1 = p_j \in (0.5,1)$ for $p = (p_j, 1-p_j-\epsilon, \epsilon)$ (i.e., $j=1$ and $\lambda=1$).

Now consider $d<0$ and $p_j < 0.5$.  As noted in \cite{Baer11}, an
application of Lyapunov's inequality for moments\cite[p.~27]{HLP}
yields $R^{d'}(p,{\mbox{\boldmath $l$}}) \leq R^d(p,{\mbox{\boldmath $l$}})$ for $d' \leq d$, and,
in particular, $R^d(p,{\mbox{\boldmath $l$}}) \leq R^0(p,{\mbox{\boldmath $l$}})$ in this case,
where
$$R^0(p,{\mbox{\boldmath $l$}}) = \sum_{i \in {\mathcal X}} p_i l_i - H_1(p)$$
via limits.  Since this is true
for all values, it is true over the minimization, and bounds for the
usual linear case apply here.  In particular, as found in
\cite{Mans} and noted in \cite{MPK}, if we define
\begin{equation}
f(p_1) = \left\{
\begin{array}{ll}
3 - 5p_1 - H_1(2p_1) & \pi_1 \leq p_1 < 0.5\\
2 - \lg 3 & 0 < p_1 < \pi_1
\end{array}
\label{p1b2}
\right.
\end{equation}
where $\pi_1 \approx 0.491$ is the root of the equality of the two terms, then this serves as an upper bound (given most probable $p_1$) on optimal redundancy (linear, and thus also $d<0$) in $(0,0.5)$.

Since this never exceeds the bound we seek here, we can now consider only
$p_j < p_1$.  Consider first those cases in which
(\ref{tightupper}) is greater than $0.5$.  In these cases, we use the
fact that $p_1 \in [p_j,1-p_j]$ to note that the maximum upper bound
over this range --- using (\ref{p1b1}) and (\ref{p1b2}) --- is
$\omega^d(p_1)$ at $p_1 = 1-p_j$, thus supplying the upper bound for
the range $(0,\pi_0^d)$, where $\pi_0^d$ is the first root of the equality 
of the two terms in the maximization at (\ref{tightupper}).

Over $p_j \in (\pi_0^d,0.5)$, we first note that $0.5$ is an upper bound via similar logic: If $p_1 \leq 0.5$, we already know that this is an upper bound.  Otherwise $p_1 \in (0.5,1-\pi_0^d)$, and (\ref{p1b1}) using $j = 1$ provides an upper bound not exceeding~$0.5$.  
\end{proof}

Fig.~\ref{fig} illustrates these bounds at a handful of values, and at limits $-1$, $0$, and~$\infty$.  For $d \rightarrow 0$, l'H\^{o}pital's rule reveals the lower bound to be the optimal one of Theorem~2 of \cite{MoAb} for $j=1$ and Theorem~4 of \cite{MPK} for arbitrary~$j$.  If one replaces optimal $\lambda$ with (possibly suboptimal) $\lceil -\lg p_j \rceil$, the upper bound becomes the suboptimal one of Lemma~1 of \cite{YeYe2}.  Taking $d \rightarrow \infty$ using, for any positive $x,y,a,b$,
$$\lim_{d \rightarrow \infty} \frac{1}{d} \lg(xa^d+yb^d) = \lg\max(a,b)$$
yields the optimal bounds of \cite{Baer11}, which are both tight.

The upper bound is clearly not optimal here, since it is not optimal for $d \rightarrow 0$ from either direction.  However, the following fact might be of help in improving this in future work:

\begin{theorem}
If $d < 0$ and $p_1 \geq 0.4$, an optimal code exists with $l_1 = 1$.
\end{theorem}

\begin{proof}
The approach here is similar to \cite{John}.  Consider the coding
step at which item $1$ gets combined with other items; we wish to
prove that this is the last step.  At the beginning of this step the
(possibly merged) items left to combine are $\{1\}, S_2^k, S_3^k,
\ldots, S_k^k$, where we use $S_j^k$ to denote the set of (individual)
items combined into a (possibly) compound item, and $w(S_j^k)$ to
denote its weight.  At this step, $p_1$ is smaller than all but
possibly one of $S_j^k$, so $(k-1)p_1 \geq (k-1)0.4$ is less than the sum of
weights, which in turn is less than or equal to~$1$.  Thus $k$ is at
most three.

Consider items $\{1\}$, $S_2^3$, and~$S_3^3$.  Assume without loss of
generality that $w(S_2^3) \geq w(S_3^3)$.  If $w(S_2^3)$ is not
compound, $\{1\}$ has the greatest weight and we are finished.  If it
is compound, call its two subtrees $S_3^4$ and $S_4^4$, in order of
nonincreasing weight.  Clearly $w(S_3^4) \leq w(S_3^3)$ due to the
combination order, so $w(S_2^3) \leq 2w(S_3^3)$.  Thus $1.5 w(S_2^3)
\leq w(S_2^3) + w(S_3^3) \leq 0.6$, so $w(S_3^3) \leq w(S_2^3) \leq
0.4$, and we can combine these two items to achieve the optimal code.
This is tight in the sense that $(p_1, (1-p_1)/3, (1-p_1)/3, (1-p_1)/3)$
has $l_1 = 2$ for $p_1 \in (0.25, 0.4)$.
\end{proof}

As an example of the improvement these bounds offer, we revisit the examples of \cite{Baer11}, which consider minimizing $L_a$ over Benford's distribution\cite{Newc,Benf}:
$$p_i = \log_{10}(i+1) - \log_{10}(i), ~ i = 1, 2, \ldots 9$$
for $a=0.6$ and $a=2$ given~$p_1$.  The bounds of \cite{Baer11} show that optimal $L_{0.6}$ for such a $p_1$ must lie in $[2.372\ldots,2.707\ldots)$.  This is identical to the application of the current result, which should not surprise, as the prior bounds apply and are tight in cases where we can show --- as in this case --- that $l_1 = 1$.  A more interesting case is that of $a=2$, for which the prior bounds, $[3.039\ldots,3.910\ldots]$, are superseded by the tighter $[3.051\ldots,3.863\ldots)$; optimal $L_2 = 3.099\ldots$.

\ifx \cyr \undefined \let \cyr = \relax \fi

\end{document}